\documentclass[letterpaper, 10 pt, conference]{ieeeconf}

\usepackage{cite}

\usepackage{amsthm}
\usepackage{amssymb}
\usepackage{amsfonts,amsmath}
\usepackage{mathtools}
\mathtoolsset{showonlyrefs}
\usepackage[font=footnotesize]{caption}
\usepackage{subcaption}
\usepackage{graphicx}
\usepackage{color, soul}
\usepackage{url}
\let\labelindent\relax % to avoid conflict of enumitem with IEEE
\usepackage{enumitem}

\usepackage[usenames,x11names,table]{xcolor}
\usepackage{xspace} % for inserting a space in TeX commands if needed
\usepackage{hyperref}
\usepackage{accents}
\usepackage{lipsum}
\usepackage[linesnumbered,vlined,ruled]{algorithm2e}
\usepackage{multirow} % for cell tables spanning multiple rows
\usepackage{array,tabularx} % for lines in tables

\newtheorem{problem}{Problem}
\newtheorem{lemma}{Lemma}
\newtheorem{theorem}{Theorem}
\newtheorem{remark}{Remark}
\newtheorem{definition}{Definition}
\newtheorem{example}{Example}
\newtheorem*{excont}{Example \continuation}
\newcommand{\continuation}{??}
\newenvironment{continueexample}[1]
 {\renewcommand{\continuation}{\ref{#1}}\excont[\textit{Cont'd}]}
 {\endexcont}
\newcommand{\xs}{\boldsymbol{\mathrm{x}}}
\newcommand{\ys}{\boldsymbol{\mathrm{y}}}

\title{\LARGE \bf Control Barrier Functions with Actuation Constraints under Signal Temporal Logic Specifications}

\author{Ali Tevfik Buyukkocak, Derya Aksaray, and Yasin Yaz{\i}c{\i}o\u{g}lu%<-this % stops a space
\thanks{A.T. Buyukkocak and D. Aksaray are with the Department of Aerospace Engineering and Mechanics, University of Minnesota, Minneapolis, MN, 55455, {\tt\small buyuk012@umn.edu, daksaray@umn.edu}, and Y. Yaz{\i}c{\i}o\u{g}lu is with the Department of Electrical and Computer Engineering, University of Minnesota, Minneapolis, MN, 55455, {\tt\small ayasin@umn.edu}
This work was \textcolor{black}{partially} supported by a MnDRIVE Graduate Research Fellowship from the University of Minnesota.}}

\IEEEoverridecommandlockouts
\begin{document}
\bibliographystyle{IEEEtran}
\maketitle
\thispagestyle{empty}
\pagestyle{empty}

\begin{abstract} \label{abstract}
We propose control barrier functions (CBFs) for a family of dynamical systems to satisfy a broad fragment of Signal Temporal Logic (STL) specifications, which may include subtasks with nested temporal operators or conflicting requirements (e.g., achieving multiple subtasks within the same time interval). The proposed CBFs take into account the actuation limits of the dynamical system as well as a feasible sequence of subtasks, and they define time-varying feasible sets of states the system must always stay inside. We show some theoretical results on the correctness of the proposed method. We illustrate the benefits of the proposed CBFs and compare their performance with the existing methods via simulations.

%The CBFs constructed with the actuation limits facilitate to accomplish subtasks expressed by a broad fragment of STL (e.g., subtasks with nested temporal operators and/or conflicting requirements such as completing multiple subtasks within overlapping time windows). We create time-varying feasible sets \textcolor{black}{of states} the system must always stay inside by incorporating feasible sequences of subtasks. We show some theoretical results for the correctness of the proposed method. We illustrate the benefits of the new CBFs and compare them with the existing methods via simulations. 
\end{abstract}

%\begin{keywords}
%abc
%\end{keywords}

%\IEEEpeerreviewmaketitle
\vspace{-1mm}
%\subsection{Dual CBF} \label{dual_CBF}
\section{Introduction} \label{sec:introduction}
Motion planning and control of dynamical systems often require the satisfaction of complex high-level specifications. Temporal logics \cite{baier2008} such as Linear Temporal Logic (LTL) \cite{pnueli1977temporal} can express such specifications and have been extensively used in the planning and control of autonomous robots (e.g., \cite{kress2009temporal,karaman,aksaray2015}). Signal Temporal Logic (STL) \cite{maler} is an expressive specification language that can define properties of dense-time real-valued signals with explicit spatial and time parameters. Control synthesis under the STL constraints has been broadly studied, and various methods have been proposed such as mixed-integer encoding of specifications (e.g., \cite{raman,buyukkocak2021acc}) and nonlinear program solutions using smooth approximations of STL robustness metrics (e.g., \cite{pant2017smooth,mehdipour2020specifying}). Alternatively, more time-efficient methods such as barrier functions are also proposed to achieve such specifications.

The authors of \cite{ames2016control} benefit from control barrier functions (CBFs) to avoid excessive reachability analysis and provide tractable feedback control laws that enable dynamical systems to stay in ``safe sets" \textcolor{black}{of states}. Moreover, they use maximal control input policies in CBFs to define such sets. Control-dependent CBFs for the sets changing with inputs are considered in \cite{ames2020integral,huang2019guaranteed}. Similarly, the forward invariance is established for time-varying sets via CBFs in \cite{xu2018constrained}. Finite-time convergent CBFs \cite{li2018formally} are also utilized to reach target sets before a preset deadline. Therefore, besides the safety critical systems, CBFs can also be used to achieve rich spatial and temporal tasks in a computationally efficient manner.

\vspace{-1mm}
\subsection{Related Work}\label{sec:related_work}
%\vspace{-1mm}
Recently, there has been a great interest in exploring the use of CBFs to satisfy temporal logic subtasks that together form the an overall specification. For example, in \cite{srinivasan2020control}, an LTL fragment is considered and a CBF is defined for each subtask while the order of the subtasks is assigned randomly. In \cite{niu2020control}, arbitrary time windows are assigned for each LTL subtask in order to avoid conflicts. That is being said, LTL specifications do not include explicit time consideration. On the other hand, specifications with explicit time windows can be expressed by STL. For instance, the authors of \cite{lindemann2018control} used time-varying CBFs to ensure the satisfaction of a family of STL specifications without considering any actuation limit. In particular, all the unsatisfied subtasks are kept active throughout the mission, and their method ensures continuous progress to the satisfaction of all subtasks. However, if two subtasks have conflicting requirements (e.g., visiting two distinct regions within overlapping time windows), their problem becomes infeasible as continuous progress to the satisfaction of these conflicting subtasks becomes infeasible. In \cite{charitidou2021barrier}, constraints on the inputs and states are considered in a receding horizon control scheme to ensure that the environment limits are not violated. However, these constraints are not explicitly used in CBF formulations, and the incapability to define the conflicting requirements remains unchanged. In \cite{gundana2021event}, an event-triggered STL framework is proposed via automata-based planning to activate CBF constraints. In the presence of the aforementioned conflicts, they notify the user and cancel the process. 

The finite and fixed time convergent CBFs are used in \cite{xiao2021high} and \cite{garg2019control} to reach target sets. While in \cite{xiao2021high} the CBFs associated with the subtasks are activated when they appear in the receding horizon, in \cite{garg2019control}, the STL subtasks at different regions are defined over distinct time windows. However, these methods do not suffice for conflicting subtasks since they may need to be active or defined at the same time. Moreover, these approaches do not allow task combinations like go to region A and/or B within \textcolor{black}{the same} time window.

The aforementioned studies generally require a continuous progress toward the target sets associated with the subtasks. Therefore, they are potentially prone to infeasibility in the presence of conflicting requirements %Therefore, when there are multiple distinct target sets that need to be visited over the overlapping time intervals 
(e.g., visiting two different regions repetitively). A common way to circumvent this is relaxing the CBF constraints (e.g., \cite{srinivasan2020control,xiao2021high,charitidou2021barrier}). However, this often yields the missing the deadline for one or more relaxed subtasks in conflict, and hence the violation of the original specification.% as will be discussed in Sec. \ref{sec:benchmark_analysis} in detail.

%In this paper, we address the CBFs for a wide variety of temporal logic specifications that may be infeasible for the existing methods without relaxing the requirements. PARAPHRASE if needed!!

% Guaranteed vehicle safety control using control-dependent barrier functions : consider safety sets that changes with control input (e.g. h(x,u)>0) 

\vspace{-2mm}
\subsection{Contributions}\label{sec:contributions}
We introduce a dual CBF formulation that can accommodate a rich fragment of STL including, but not limited to, nested temporal operators and specifications requiring to complete multiple distinct subtasks in overlapping time windows (e.g., eventually $x\geq1$ and $x\leq-1$ within $t \in [0,5]$ with $x_0=0$). In the proposed CBF formulation, the actuation limits, as well as a feasible sequence of subtasks, are considered to define the largest time varying sets of states that we want our system to stay inside. This does not prohibit the system to move away from the active target areas, unlike the continuous progress toward them which is often required in the related work. Instead of considering each subtask one by one (or when they appear in a receding horizon) we determine a feasible sequence of subtasks whose achievement yields the satisfaction of the STL specification. %This may also yield the trajectories with early satisfaction and reduced input cost.
\vspace{-5mm}
\section{Preliminaries}\label{sec:preliminaries}
In this paper, $\mathbb{R}_{\ge0}$ refers to the set of nonnegative real numbers, and $\mathbb{R}^n$ denotes the set of $n$-dimensional real-valued vectors. We represent the set cardinality and the Euclidean norm of a vector by $|\cdot|$ and $\|\cdot\|$, respectively. The relations between vectors (e.g., $\geq,\leq$) are implemented elementwise. The class $\mathcal{K}$ functions, $\alpha:\mathbb{R}_{\geq0}\to\mathbb{R}_{\geq0}$, are strictly increasing, i.e., $\alpha(x)>\alpha(y)$ if $x>y$, with $\alpha(0)=0$. A function $f: \mathcal{X} \subseteq \mathbb{R}^n \to \mathbb{R}^n$ is locally Lipschitz continuous at $\mathcal{X}$ if every point on $\mathcal{X}$ has a neighborhood $\mathcal{N}$ and there exists a constant $L>0$ such that $\|f(\xs)-f(\ys)\|\leq L\|\xs-\ys\|$ for all $\xs,\ys \in \mathcal{N}$. A set is compact if it is closed and bounded. Let the state and input of a dynamical system be denoted by $\xs\in\mathcal{X}$ and $\boldsymbol{u}\in\mathcal{U}$, respectively, where $\mathcal{X}\subseteq\mathbb{R}^n$ and $\mathcal{U}\subseteq\mathbb{R}^m$ are the state space and the admissible set of control, respectively. We use the subscript $(\cdot)_t$ to denote the value of a term at time $t$. Moreover, a control affine system is given as,
%\vspace{-1mm}
\begin{equation}
    \small
    \label{eq:cont_affine}
    \dot{\xs}=f(\xs)+g(\xs)\boldsymbol{u},
    %\vspace{-1mm}
\end{equation}
\noindent with locally Lipschitz $f:\mathbb{R}^n\to\mathbb{R}^n$ and $g:\mathbb{R}^n\to\mathbb{R}^{n\times m}$. 

\vspace{-1mm}
\begin{definition}\label{def:forward_invariant}
(Forward Invariance). A set $\mathcal{C}\in\mathbb{R}^n$ is forward invariant with respect to \eqref{eq:cont_affine} if for every $\xs_0 \in \mathcal{C},\ \xs_t \in\mathcal{C},\ \forall t$.% for all $t$.
\end{definition}
\vspace{-2mm}

\subsection{Time-varying (Zeroing) Control Barrier Functions} \label{sec:TVZCBF}

Control barrier functions (CBFs) are used in \cite{ames2016control} to render the desired set of states forward invariant. That is, let $h_{CBF}:\mathcal{X}\subseteq\mathbb{R}^n\to\mathbb{R}$ be a differentiable function. If the system starts in the zero superlevel set $\mathcal{C}=\{\xs\in\mathcal{X}\ |\ h_{CBF}(\xs)\geq0 \}$, i.e, $\xs_0\in\mathcal{C}$, then the CBF constraints force the system to stay always inside the set $\mathcal{C}$.

The author of \cite{xu2018constrained}
uses time varying CBFs to ensure forward invariance in sets changing with time, i.e., $\mathcal{C}_t$. In this case, a differentiable function $\mathfrak{h}(\xs,t):\mathcal{X}\times I\to\mathbb{R}^n$ is defined as a time-varying control barrier function in the time domain $I\subseteq\mathbb{R}_{\geq0}$. The time-varying set $\mathcal{C}_t=\{\xs\in\mathcal{X}\ |\ \mathfrak{h}(\xs,t)\geq 0\}$  %\textcolor{red}{$\mathcal{C}=\{\xs\in\mathcal{X}\ |\ \mathfrak{h}(\xs,r)\geq 0\}$}
is rendered forward invariant by the input $\boldsymbol{u}\in\mathcal{U}$, if $\xs_0\in\mathcal{C}_0$ and there exists a locally Lipschitz continuous class $\mathcal{K}$ function $\alpha:\mathbb{R}_{\geq 0}\to\mathbb{R}_{\geq0}$ such that 
\begin{equation}\label{eq:TVZCBF}
%\vspace{-1mm}
\small
    \sup_{\boldsymbol{u}\in\mathcal{U}} \frac{\partial\mathfrak{h}(\xs,t)}{\partial\xs}^T(f(\xs)+g(\xs)\boldsymbol{u})+\frac{\partial\mathfrak{h}(\xs,t)}{\partial t}+\alpha(\mathfrak{h}(\xs,t))\geq0,
    %\vspace{-1mm}
\end{equation}
\noindent then $\xs_t\in\mathcal{C}_t,\ \forall t$. \textcolor{black}{In the remainder of this paper, we use $\mathcal{C}$ to denote $\mathcal{C}_t$ for the sake of simplicity, and $\mathcal{C}_i$ will denote the feasible set associated with the $i^{th}$ subtask.}

\vspace{-1mm}
\subsection{Signal Temporal Logic} \label{sec:stl}
Signal Temporal Logic (STL) \cite{maler} can express rich properties of time series, and we specify desired system behaviors with the following \textcolor{black}{expressive} STL syntax:
\vspace{-1mm}
\begin{equation}
\small
\setlength{\jot}{.8pt}
\label{eq:STL_fragment}
\begin{split}
\Phi &= \Phi_1\wedge\Phi_2\ |\ {\Phi_1\vee\Phi_2}\ |\ F_{[a,b]}\phi\ |\ G_{[a,b]}\phi\ |\ {\varphi_1U_{[a,b]}\varphi_2},\\
\phi &= \varphi\ |\ \neg\phi\ |\ F_{[c,d]}\varphi\ |\ G_{[c,d]}\varphi,\\
\varphi &= \mu\ |\ \neg\varphi \ |\ \varphi_1\wedge\varphi_2 \ |\ \varphi_1\vee\varphi_2,
\end{split}   
\end{equation}
\vspace{-2mm}

\noindent where $a,b,c,d\in\mathbb{R}_{\ge0}$ are time bounds with $b\geq a$, $d\geq c$\textcolor{black}{; $F_{[a,b]}$, $G_{[a,b]}$, $U_{[a,b]}$, $\wedge$, $\vee$, $\neg$ are finally (eventually), globally (always), until, conjunction, disjunction, and negation operators, respectively}; $\Phi$ is an STL specification, and $\mu$ is a predicate in inequality form, $\mu=p(\xs_t)\sim0$, with signal $\xs:\mathbb{R}_{\ge0} \to\mathbb{R}^n$, function $p:\mathbb{R}^n \to \mathbb{R}$, and $\sim\in\{\geq,\leq\}$. While the fragment in \eqref{eq:STL_fragment} allows for nested temporal operators, temporal operators cannot be in conjunction/disjunction with bare predicates. Still, the fragment is more expressive than that of other STL CBF literature (e.g., \cite{lindemann2018control}, \cite{xiao2021high}).

Let $\xs_t$ denote the value of $\xs$ at time $t$, and $\xs_{t',t''}$ represents the trajectory of the system within $[t',t'']\subseteq\mathbb{R}_{\geq0}$. %let  and $(\xs,t)$ be the part of the signal that is a sequence of $\xs_{t'}$ for $t'\in[t,\infty)$. 
The satisfaction of an STL specification by the part of the signal starting at $t$, i.e., $\xs_{t,T}$, is determined as follows:
\vspace{-1mm}
\begin{equation}
\small
\begin{split}
\xs_{t,T} &\vDash \mu \Longleftrightarrow p(\xs_t)\sim\, 0, \\
\xs_{t,T} &\vDash \neg\mu \Longleftrightarrow \neg\big(\xs_{t,T}\vDash \mu\big),\\
\xs_{t,T} & \vDash \Phi_1\wedge\Phi_2 \Longleftrightarrow \xs_{t,T}\vDash \Phi_1 \;\text{and}\; \xs_{t,T}\vDash \Phi_2,\\
\xs_{t,T} & \vDash \Phi_1\vee\Phi_2 \Longleftrightarrow \xs_{t,T}\vDash \Phi_1 \;\text{or}\; \xs_{t,T}\vDash \Phi_2,\\
\xs_{t,T} & \vDash \Phi_1 U_{[a,b]}\Phi_2\Longleftrightarrow \exists t'\in[t+a,t+b],\;\xs_{t',T}\vDash \Phi_2 \\&\hspace{2.5cm}\text{and}\; \forall t''\in[t,t'],\;\xs_{t'',T}\vDash \Phi_1, \\
\xs_{t,T} & \vDash G_{[a,b]}\Phi\Longleftrightarrow \forall t'\in[t+a,t+b],\;\xs_{t',T}\vDash \Phi, \\
\xs_{t,T} & \vDash F_{[a,b]}\Phi\Longleftrightarrow \exists t'\in[t+a,t+b],\;\xs_{t',T}\vDash \Phi.
\end{split}
\label{eq:stl_semantics}
\end{equation}

While $F_{[a,b]}\Phi$ implies that $\Phi$ must be true at least once within $[t+a,t+b]$, $G_{[a,b]}\Phi$ requires the continuous satisfaction of $\Phi$ within the same interval. Until operator $\Phi_1 U_{[a,b]}\Phi_2$, on the other hand, forces $\Phi_1$ to be true until $\Phi_2$ is achieved which happens at least once within $[t+a,t+b]$.

\textcolor{black}{The horizon of an STL specification $\Phi$, i.e., $hrz(\Phi)$, is defined as the minimum amount of time required to decide whether the formula is satisfied \cite{dokhanchi}. For instance, the formula $F_{[0,6]}G_{[0,3]}x\geq0$ has a horizon of $6+3=9$, and the formula $G_{[0,15]} x\geq 0 \wedge F_{[0,8]} x \geq 10$ has a horizon of $\max(15,8)=15$. The inductive equations to compute the horizon of an STL specification can be found in \cite{dokhanchi}. We consider the mission horizon to be equal to STL horizon as $T=hrz(\Phi)$.} 

\begin{definition}\label{def:DNF}
(Disjunctive Normal Form (DNF) \cite{baier2008}). 
A specification is said to be in DNF if it has the form of $\bigvee_i\bigwedge_j\mu_{i,j}$ as the disjunction of conjunctions of predicates.
\end{definition}

In this paper, the desired STL specification $\Phi$ expresses a set of subtasks \textcolor{black}{$\Phi_i$}, \textcolor{black}{each of which includes inner specifications $\varphi$ in DNF, according to \eqref{eq:STL_fragment}} as follows:

\begin{equation}
\small
\label{eq:structure}
    \Phi=\Phi_1\wedge\Phi_2\wedge\ldots\wedge\Phi_k.
    \vspace{-1mm}
\end{equation}

\begin{remark}\label{rem:negation_free}
Any STL specification can be rewritten in negation-free \textcolor{black}{form} \cite{ouaknine2008some}. Furthermore, we can rewrite any inner specification $\varphi_i$ in \eqref{eq:STL_fragment} in DNF \cite{baier2008}.
\end{remark}

For example, for the specification $\Phi_i=\neg G_{[0,T]} \big(x\leq 0\vee(y\leq0\wedge z\leq 0)\big)$ we denote its inner part as $\varphi_i=x\leq 0\vee(y\leq0\wedge z\leq 0)$, i.e., $\Phi_i=\neg G_{[0,T]}\varphi_i$. The specification can be rewritten as $\Phi_i=F_{[0,T]} \big((x\geq 0\wedge y\geq 0)\vee(x\geq 0\wedge z\geq 0)\big)$. In this study, we consider the STL specifications formed by the conjunction of such subtasks, $\Phi_i$, as in \eqref{eq:structure}.

\section{Problem Formulation}\label{sec:problem_formulation}
We use a family of dynamical systems such that $f(\xs)=\boldsymbol{0_n}$ and $g(\xs)=\boldsymbol{I}_{n}$ where $\boldsymbol{0_n}$ and $\boldsymbol{I}_{n}$ denote $n\times1$ vector of zeros and $n\times n$ identity matrix, respectively. That is, we require the system to be holonomic with equal number of controllable and total degrees of freedom, i.e., $n=m$. This renders the system 
\begin{equation}\label{eq:single_integrator}
\small
    \Dot{\xs}=\boldsymbol{u}.
\end{equation}
We also consider the admissible input set as 
\begin{equation} \label{eq:input_set}
\small
    \mathcal{U}=\{\boldsymbol{u}\in\mathbb{R}^n\ |\ \|\boldsymbol{u}\|\leq u_{\max}\},
\end{equation} 

\noindent where ${u}_{\max}\in\mathbb{R}_{\geq0}$ is the scalar maximum input.

%In this paper, we address the following optimal control problem considering the STL specifications as in \eqref{eq:structure}.
\vspace{-1mm}
\begin{problem}\label{problem_1}
Given a dynamical system \eqref{eq:single_integrator}, find an optimal control law, $\boldsymbol{u}_{0,T}^*$, such that the resulting system trajectory $\xs_{0,T}$ satisfies the given STL specification $\Phi$ \textcolor{black}{with the minimum control effort}, as long as such a satisfaction is feasible with the initial state $\xs_0$ and the control bounds $\boldsymbol{u}\in\mathcal{U}$.
\end{problem}

\vspace{-3mm}
\section{Solution Approach}\label{sec:solution_approach}
The method in \cite{lindemann2018control} uses time-varying CBFs to satisfy \textcolor{black}{a limited fragment of STL specifications and do not consider} the actuation constraints. Our method also uses the time-varying CBFs but can accommodate richer STL specifications including, but not limited to, the ones requiring recurrence \textcolor{black}{via nested temporal operators} and/or containing conflicting requirements. 

We consider an STL specification $\Phi$ formed by the conjunction of subtasks $\Phi_i$ as in \eqref{eq:structure}. \textcolor{black}{We assume that} each subtask is associated with compact target set(s) of states. We will assign a primary CBF for each subtask $\Phi_i$ based on the actuation limits to keep the system in a feasible time-varying set $\mathcal{C}_i$ to achieve the inner specification \textcolor{black}{$\varphi$ on time \big(e.g., $\Phi_i=F_{[a,b]}(\varphi_i\vee\varphi_j)$\big)}. We will also introduce another (secondary) CBF incorporating all the subtasks that are not achieved yet (i.e., dual CBF formulation).
%This secondary CBF enables to determine the system trajectory such that a subtask $\Phi_i$ in \eqref{eq:structure} is realized by accounting for the achievability of the others ($\Phi_j$, $j\in\{1,\ldots,k\}\setminus \{i\}$) as well. 
To ensure the reachability of the target set of each subtask, the secondary CBF will be again constructed based on the actuation limits.

\vspace{-1mm}
 \subsection{Control Barrier Functions with Actuation Limits}\label{sec:CBF_act}

\begin{definition}\label{def:signed_distance}
(Signed Distance \cite{vandenberghe2004convex}). Let $\xs\in\mathcal{X}$ be a point in metric space $\mathcal{X}$ and $\mathcal{P}\subseteq\mathcal{X}$ be a compact set in which the inner subtask $\varphi$ defined in \eqref{eq:STL_fragment} is satisfied, i.e., $\mathcal{P}=\{\xs\in\mathcal{X}\ \vert\ \xs \models \varphi\}$. Then, we define the
\begin{itemize}
    \item Distance from $\xs$ to $\mathcal{P}$ as $dist(\xs,\mathcal{P}):=\inf\limits_{\xs^*\in \mathcal{P}}\| \xs-\xs^*\|$,
    \item Depth of $\xs$ in $\mathcal{P}$ as $depth(\xs,\mathcal{P}):=dist(\xs,\mathcal{X}\setminus \mathcal{P})$,
    \item Signed Distance from $\xs$ to $\mathcal{P}$ to be
    \vspace{-1mm}
     \begin{equation}
     \small
     \label{eq:signed_dist}
 Dist(\xs,\mathcal{P}):=\left\{
\begin{array}{rl}
      depth(\xs,\mathcal{P}), & \textit{if}\ \xs\in \mathcal{P}, \\
      -dist(\xs,\mathcal{P}), & \textit{otherwise}. \\
\end{array} 
\right.
 \end{equation}
\end{itemize}
\end{definition}
\vspace{-1mm}
We adopt the robustness metric definition in \cite{fainekos2009robustness} which uses the distance to the target sets in the state space \eqref{eq:signed_dist}. We will calculate the robustness metrics pointwise over the time (i.e., static robustness) and use the CBFs to achieve the subtasks on time\footnote{\textcolor{black}{The distance function in Def. \ref{def:signed_distance} may not be differentiable at some point inside the target set $\mathcal{P}$. However, the subtask removal discussed in the next sections will prevent the system from reaching to such point similar to switch mechanism in \cite{lindemann2018control}.}} 
Several studies use the smooth approximations of $\min$ and $\max$ functions to define STL robustness metrics with (e.g., \cite{gilpin2020smooth}) and without (e.g., \cite{aksaray2016q,pant2017smooth}) soundness guarantee.
We use the following approximations to preserve the soundness by lower bounding the $\min$ and $\max$ functions:
\vspace{-2mm}
\begin{definition} \label{def:smooth_approx}
(Smooth approximations of $\min$ and $\max$ functions \cite{vandenberghe2004convex}).
\begin{equation}\label{eq:smooth_minmax}
\small
    \begin{split}
        \min_i \boldsymbol{c}(i)&\geq\widetilde{\min_i}\, \boldsymbol{c}(i)=-\frac{1}{\beta}ln\big(\sum\nolimits_i e^{-\beta \boldsymbol{c}(i)}\big),\\
        \max_i \boldsymbol{c}(i)&\geq\widetilde{\max_i}\, \boldsymbol{c}(i)=\frac{\sum\nolimits_i \boldsymbol{c}(i) e^{\beta \boldsymbol{c}(i)}}{\sum\nolimits_i e^{\beta \boldsymbol{c}(i)}}.\\
    \end{split}
\end{equation}
\noindent where \textcolor{black}{$\beta\in\mathbb{R}^+$} is a tuning coefficient whose higher values yield better approximations.
\end{definition}
\vspace{-1mm}
 For each subtask $\Phi_i$ in \eqref{eq:structure}, the inner part, $\varphi_i$, may include only disjunctions and conjunctions based on \eqref{eq:STL_fragment} and Remark \ref{rem:negation_free}. Using the smooth approximations of $\min$/$\max$ functions, we next define the static robustness semantics for the inner specifications, i.e., $\varphi$ in \eqref{eq:STL_fragment}, as \cite{fainekos2009robustness},
\begin{equation}
\small
\begin{split}
&\rho(\xs_t,\varphi)=Dist(\xs_t,\mathcal{P}), \\
&\rho(\xs_t,\varphi_{1}\wedge\varphi_{2}) =\widetilde{\min}\big(\rho(\xs_t,\varphi_{1}),\rho (\xs_t,\varphi_{2}) \big), \\
&\rho(\xs_t,\varphi_{1}\vee\varphi_{2}) =\widetilde{\max} \big(\rho(\xs_t,\varphi_{1} ) ,\rho  (\xs_t, \varphi _{2})\big). \\
\end{split}
\label{eq:robustness_degree}
\end{equation}

Note that \textcolor{black}{$\rho(\xs_t,\varphi)\geq0\Longleftrightarrow \xs_t\models\varphi$ and,} $\varphi$ can include a single predicate which depicts a compact set (e.g., $\varphi=\mu=x^2+y^2\leq r^2$) or conjunction of the predicates that can only form a compact set when applied together (e.g., $\varphi=x\geq3\wedge x\leq 5$). 

\textcolor{black}{Using the smooth approximations of $\min$ and $\max$ functions allows for the solution of Problem \ref{problem_1} via gradient-based techniques. Note that we want to see the robustness metric nonnegative (i.e., $\rho(\xs_t,\varphi)\geq0\Longleftrightarrow \xs_t\models\varphi$) in accordance with the outer temporal operators, e.g., always when the operator is globally ($G_{[a,b]}\varphi$) and at least once when it is finally ($F_{[a,b]}\varphi$). We will enforce the satisfaction of $\rho(\xs_t,\varphi)\geq0$ within the particular time windows, e.g., $t\in[a,b]$, via CBFs.}

Our time-varying CBF definitions depend on the remaining time $r$ \textcolor{black}{to achieve each subtask $\Phi_i$} (strictly decreasing with time, i.e., $r_{t}=r_{0}-t$). In this regard, we define a CBF for each subtask such that the system will always remain in a distance with the target set of states that can be visited ($\rho(\xs,\varphi_i)\geq0$ can be achieved) in at most $r_i$ time units. For the system \eqref{eq:single_integrator}, by using the bounding input ${u}_{\max}\in\mathbb{R}_{\geq0}$ in \eqref{eq:input_set} we can define a primary CBF for each subtask $\Phi_i$ as
\begin{equation}
\small
\label{eq:primary_CBF}
    \mathfrak{h}_i(\xs,t)=r_i+\frac{\rho(\xs,\varphi_i)}{{u}_{\max}},
    \vspace{-1mm}
\end{equation}

\noindent where we want our system to stay inside the set $\mathcal{C}_i=\{\xs\in\mathbb{R}^n\ |\ \mathfrak{h}_i(\xs,t)\geq 0\}$ all the time until $\Phi_i$ is satisfied.

\begin{example} Let an STL specification be given as
$\Phi=\Phi_1\wedge\Phi_2$ with $\Phi_1=F_{[0,5]}\big((x\geq2\wedge x\leq 3) \vee (x\geq5 \wedge x \leq6)\big)$ and $\Phi_2=G_{[10,15]}(x\geq 7 \wedge x\leq8)$ for $x\in\mathbb{R}$. Then $\rho(x,\varphi_1)=\widetilde{\max}\big(\widetilde{\min}(x-2,3-x),\widetilde{\min}(x-5,6-x)\big)$ and $\rho(x,\varphi_2)=\widetilde{\min}(x-7,8-x)$. We construct the primary CBFs for $\Phi_1$ and $\Phi_2$ as
$\mathfrak{h}_1(x,t)=r_1+\frac{\rho(x,\varphi_1)}{{u}_{\max}}$ and $\mathfrak{h}_2(x,t)=r_2+\frac{\rho(x,\varphi_2)}{{u}_{\max}}$, respectively, where $r_{1_0}=5$ and $r_{2_0}=10$.
\label{example_1}
\end{example}

Notice that $\mathcal{C}_i$ is constructed with more information than the target set, $\mathcal{P}$, by incorporating both the actuation constraints and the remaining time to reach the target. This enables the system to move freely, even away from the active target sets. Accordingly, the system can satisfy more complex specifications compared to the related work (e.g., \cite{lindemann2018control,xiao2021high}) which generally require continuous progress toward the target. For instance, we can specify conflicting requirements, e.g., $\Phi=F_{[0,5]}(x\geq-1 \wedge x\leq 0)\wedge F_{[0,5]}(x\geq3\wedge x\leq 4)$, that require movement to opposite directions (e.g., when $0<x_0<3$) as we discuss further in Sec. \ref{sec:benchmark_analysis}. 

\textcolor{black}{The only exception to the primary CBF definition in \eqref{eq:primary_CBF} is the disjunction of temporal operators. In this case, we define the primary CBF associated with the disjunction using the smooth approximation of $\max$ function in \eqref{eq:smooth_minmax} as follows:
\begin{equation}
\small
    \label{eq:primary_CBF_OR}
         \Phi_{\vee}=\bigvee\nolimits_i\Phi_i\Longleftrightarrow\mathfrak{h}_{\vee}(\xs,t)=\widetilde{\max}_i\,\mathfrak{h}_i(\xs,t),
\end{equation}
\noindent and $\Phi_{\vee}$ is treated as any other subtask in \eqref{eq:structure}.\\ 
\noindent\textbf{Extension of the approach:} The system \eqref{eq:single_integrator} with a scalar input bound $u_{\max}$ guarantees that if the primary CBF \eqref{eq:primary_CBF} is initially nonnegative, then the system can reach the target set on time. Note that the bound $u_{\max}$ does not change with time. That is, ensuring $\mathfrak{h}(\xs,t)\geq0$ within $t\in[0,r]$, forces the system to stay in the proximity of the target set such that in the worst case the set can be reached in at most $r-t$ time units (and is reached when $t=r$) with the maximal input policies. On the other hand, for more complex systems than \eqref{eq:single_integrator}, e.g., \eqref{eq:cont_affine}, we could have nonnegative CBFs over the time, if we had an overall bound on the derivative of each state and a system property making this bound nondecreasing with time (at least under maximal inputs). Then the primary CBF could preserve its nonnegativeness. For example, monotone and/or cooperative control systems \cite{angeli2003monotone}, can satisfy this under mild assumptions.} 

\vspace{-1mm}
\subsection{Control Barrier Functions with Sequential Subtasks}
In the presence of $k>1$ subtasks in conjunction as in \eqref{eq:structure}, the system should lie in the intersection of the superlevel sets of barrier functions as $\xs\in\mathcal{C}$ where $\mathcal{C}=\{\xs\in\mathbb{R}^n\ |\  \mathfrak{h}_i(\xs,t)\geq 0, \forall i\in\{1,\ldots,k\}\}$, i.e., $\mathcal{C}=\bigcap_i\mathcal{C}_i$ as emphasized in the STL CBF literature (e.g., \cite{xiao2021high,gundana2021event}). However, since each feasible set $\mathcal{C}_i$ shrinks with time, subtasks with conflicting requirements within the overlapping time windows may cause an empty intersection of the feasible sets. The relaxation of CBF constraints to avoid infeasibility can also cause violation of the \textcolor{black}{relaxed} specifications. %Two regions that need to be visited eventually within the same time interval but placed at opposite directions with respect to the system may be given as an example of such a case. 
For this reason, we propose a dual CBF formulation and introduce a secondary CBF to ensure that any $\Phi_i$ is achieved by taking into account the shrinkage of feasible sets associated with the remaining subtasks $\Phi_j$, $\forall j\in\{1,\ldots,k\}\setminus \{i\}$. %In other words, while achieving the $i^{th}$ subtask, the system is aware that there will be others after $i$. 
Therefore, even if there is more time to achieve it, the system may go for $\Phi_i$ immediately in accordance with the remaining subtask deadlines. In the secondary CBF formulation, we utilize a sequence of subtasks which is feasible to achieve in that particular order (even it can be disobeyed during execution). A term in the sequence is removed once the associated subtask is satisfied. Next, we elaborate more on the construction of the sequence of subtasks.
\vspace{-2mm}
\begin{definition} \label{def:set_distance}
(Ordered Set Distance\footnote{\textcolor{black}{This metric is a form of Hausdorff distance with a particular order.}}). We define the distance between two compact sets $\mathcal{P}_i,\, \mathcal{P}_j$ with the given order as
\begin{equation}
\small
\label{eq:set_set_dist}
         Dist(\mathcal{P}_i,\mathcal{P}_j)=\sup\limits_{\xs\in \mathcal{P}_i}\ dist(\xs,\mathcal{P}_j).
\end{equation}
\end{definition}
\vspace{-2mm}
Intuitively, the worst-case scenario is considered in the definition of the set distance where the system starts at the farthermost point in $\mathcal{P}_i$ and travels to $\mathcal{P}_j$. For notational simplicity, by using the Defs. \ref{def:signed_distance} and \ref{def:set_distance} together with the input bound $u_{\max}$ and with a slight abuse of notation, we define 
\vspace{-1mm}
\begin{equation} \label{eq:d_dist}
\small
    \mathfrak{d}_i=-Dist(\xs,\mathcal{P}_i)/{u}_{\max}\ \text{and}\
    \mathfrak{d}_{i,j}=Dist(\mathcal{P}_i,\mathcal{P}_j)/{u}_{\max},
\vspace{-1mm}
\end{equation}
\noindent as the required time for the system to reach $i^{th}$ set under the maximum input and as the worst-case minimum travel time between \textcolor{black}{$i^{th}$ and $j^{th}$ sets}, respectively. 
\vspace{-2mm}
%That is, $\mathfrak{d}_i$ is the earliest time to achieve the subtask $\Phi_i$, and $\mathfrak{d}_{i,j}$ is the worst-case minimum travel time between the sets of states in which the subtasks $\Phi_i$ and $\Phi_j$ are achieved. Next we show that subsequences of a feasible sequence preserve the feasibility along the mission when a term is removed (even if it is a term later than the first one).
\begin{definition}\label{def:sequence}
(Subtask Sequence and Its Subsequences). Given an STL specification with $k$ subtasks as in \eqref{eq:structure}, the subtask sequence $\boldsymbol{S}$ is an order of the subtask indices $\{1,\ldots,k\}$ with $\vert \boldsymbol{S}\vert=k$. For example, consider $\Phi=\Phi_1\wedge\Phi_2\wedge\Phi_3\wedge\Phi_4$, then $\boldsymbol{S}=(1,4,3,2)$ is a subtask sequence for $\Phi$. %Define a function $s(\cdot):\boldsymbol{S}\to\{1,\ldots,k\}$ that maps the order in the sequence to the subtask index, e.g., $s(2)=4$. 
Each term in the sequence $\boldsymbol{S}(i)$ has a corresponding compact target set $\mathcal{P}_{\boldsymbol{S}(i)}$ in which the inner part $\varphi_{\boldsymbol{S}(i)}$ of the subtask $\Phi_{\boldsymbol{S}(i)}$ is satisfied. A subsequence $\boldsymbol{S}'$ is obtained by removing one or more terms of $\boldsymbol{S}$ after the associated target sets are visited where $1\!\leq\!|\boldsymbol{S}'|\!\leq\! k$, and $|\boldsymbol{S}'|\!=\! k$ implies $\boldsymbol{S}'\!=\!\boldsymbol{S}$. For example, $(4,3,2)$, $(1,3,2)$, and $(1,3)$ are some subsequences of $\boldsymbol{S}$.
% subsequences preserve the order
\end{definition}
\vspace{-1mm}
Let us start with a sequence of subtasks whose accomplishment in the given order implies the satisfaction of the STL specification. We next show that even if the order in $\boldsymbol{S}$ is violated during the execution (e.g., satisfying a later subtask in the sequence first and removing it), the resultant subsequence will still be feasible. %\textcolor{black}{The proofs of our theoretical results can be found in the extended version of the paper in \cite{buyukkocakarxiv}.}
\vspace{-1mm}
\begin{lemma} % the order matters in sequences!!!!
\label{lem:triangle} Consider an STL specification $\Phi=\Phi_1\wedge\ldots\wedge\Phi_k$ and a sequence of its subtasks $\boldsymbol{S}$ as in Def. \ref{def:sequence} together with the duration definitions in \eqref{eq:d_dist}.
%Let a sequence of $k$ subtasks in \eqref{eq:structure} be given as $\boldsymbol{S}$ and its subsequences $\boldsymbol{S}'\textcolor{black}{\subset} \boldsymbol{S}$ be created by removing the terms of $\boldsymbol{S}$ after a set associated with a particular term (subtask) is visited. Define the time required to travel to/between the sets as $\mathfrak{d}_i=-Dist(\xs,\mathcal{P}_i)/{u}_{\max}$ and $\mathfrak{d}_{i,j}=Dist(\mathcal{P}_i,\mathcal{P}_j)/{u}_{\max}$, respectively, via Defs. \ref{def:signed_distance} and \ref{def:set_distance}. 
If $ r_{\boldsymbol{S}(i)}\geq \mathfrak{d}_{{\boldsymbol{S}(1)}} + \sum_{l=2}^{i} \mathfrak{d}_{{\boldsymbol{S}(l-1)},{\boldsymbol{S}(l)}},\ \forall i\in \boldsymbol{S}$, then $ r_{{\boldsymbol{S}'(j)}}\geq \mathfrak{d}_{{\boldsymbol{S}'(1)}} + \sum_{l=2}^{j} \mathfrak{d}_{{\boldsymbol{S}'(l-1)},{\boldsymbol{S}'(l)}},\ \forall j\in \boldsymbol{S}'$, for any subsequence $\boldsymbol{S}'$.
\end{lemma}
\begin{proof}
\textbf{Case 1:} (Only the first term in the sequence, $\mathfrak{d}_{\boldsymbol{S}(1)}$, is removed, i.e., $\boldsymbol{S}'=\boldsymbol{S}\setminus\boldsymbol{S}(1)$.) The term $\mathfrak{d}_{\boldsymbol{S}(1)}$ is removed when the set of states associated with it, $\mathcal{P}_{\boldsymbol{S}(1)}$, is visited. We want to show that after the removal of $\boldsymbol{S}(1)$, $r_{\boldsymbol{S}(i)}\!\geq\! \mathfrak{d}_{\boldsymbol{S}(2)}\! +\! \sum_{l=3}^i \mathfrak{d}_{\boldsymbol{S}(l-1),\boldsymbol{S}(l)},\forall i\in \boldsymbol{S}\setminus\boldsymbol{S}(1)$ holds. It is given in the premise that $ r_{\boldsymbol{S}(i)}\!\geq\! \mathfrak{d}_{{\boldsymbol{S}(1)}}\! +\! \sum_{l=2}^{i} \mathfrak{d}_{{\boldsymbol{S}(l-1)},{\boldsymbol{S}(l)}}, \forall i\in \boldsymbol{S}$, then we have $ r_{\boldsymbol{S}(i)}\geq \sum_{l=2}^{i} \mathfrak{d}_{{\boldsymbol{S}(l-1)},{\boldsymbol{S}(l)}}=\mathfrak{d}_{{\boldsymbol{S}(1),\boldsymbol{S}(2)}}+\sum_{l=3}^{i} \mathfrak{d}_{{\boldsymbol{S}(l-1)},{\boldsymbol{S}(l)}}, \forall i\in \boldsymbol{S}$. Moreover, by the Defs. \ref{def:signed_distance} and \ref{def:set_distance}, we know $\mathfrak{d}_{\boldsymbol{S}(1),\boldsymbol{S}(2)}\!\geq\!\mathfrak{d}_{\boldsymbol{S}(2)}\!=\!-Dist(\xs,\mathcal{P}_{\boldsymbol{S}(2)})/{u}_{\max}, \forall\xs\in\mathcal{P}_{\boldsymbol{S}(1)}$. Hence, the condition in the premise implies $r_{\boldsymbol{S}(i)}\geq \mathfrak{d}_{\boldsymbol{S}(2)} + \sum_{l=3}^i \mathfrak{d}_{\boldsymbol{S}(l-1),\boldsymbol{S}(l)}, \forall i\in \boldsymbol{S}\setminus\boldsymbol{S}(1)=\boldsymbol{S}'$.\\
\textbf{Case 2:} (Later subtasks in the sequence are removed.) Let $\boldsymbol{S}_{-1}'$ be any subsequence of $\boldsymbol{S}\setminus\boldsymbol{S}(1)$ and $\boldsymbol{S}_{-1}''$ be any subsequence of $\boldsymbol{S}_{-1}'$. Then it follows from the triangle inequality that $\sum_{l=2}^{\vert\boldsymbol{S}_{-1}'\vert} \mathfrak{d}_{\boldsymbol{S}_{-1}'(l-1),\boldsymbol{S}_{-1}'(l)}\geq\sum_{l=2}^{\vert\boldsymbol{S}_{-1}''\vert} \mathfrak{d}_{\boldsymbol{S}_{-1}''(l-1),\boldsymbol{S}_{-1}''(l)}$, which implies $\mathfrak{d}_{{\boldsymbol{S}(1)}} +\sum_{l=2}^{\vert\boldsymbol{S}_{-1}'\vert} \mathfrak{d}_{\boldsymbol{S}_{-1}'(l-1),\boldsymbol{S}_{-1}'(l)}\geq\mathfrak{d}_{{\boldsymbol{S}(1)}} +\sum_{l=2}^{\vert\boldsymbol{S}_{-1}''\vert} \mathfrak{d}_{\boldsymbol{S}_{-1}''(l-1),\boldsymbol{S}_{-1}''(l)}$. Since $\boldsymbol{S}_{-1}'$ and $\boldsymbol{S}_{-1}''$ are arbitrary, the conditional statement in the premise holds.\\
\textbf{Case 3:} (The first and later subtask(s) are removed together.) It follows from the Cases 1 and 2 that any subsequence obtained by removing $\boldsymbol{S}(1)$ and one or more later subtasks in $\boldsymbol{S}$ together preserves the inequality in the premise.
%It directly follows from the triangle inequality, e.g., $\mathfrak{d}_{\boldsymbol{S}(i),\boldsymbol{S}(k)}\leq \mathfrak{d}_{\boldsymbol{S}(i),\boldsymbol{S}(j)}+\mathfrak{d}_{\boldsymbol{S}(j),\boldsymbol{S}(k)}$ for $i>j>k$. 
\end{proof}

%\textcolor{red}{\begin{definition}\label{def:portion}
%(Portion of a Sequence). We define the portion of a sequence $\boldsymbol{S}=(1,\ldots,\vert\boldsymbol{S}\vert)$ by $\boldsymbol{S}^\gamma=\boldsymbol{S}\setminus(\gamma+1,\ldots,\vert\boldsymbol{S}\vert)$, i.e., $\boldsymbol{S}^\gamma=(1,\ldots,\gamma)$, where $2\leq\gamma\leq \vert\boldsymbol{S}\vert$.
%\end{definition}}
\vspace{-1.2mm}
We construct the secondary CBF candidates over a feasible sequence $\boldsymbol{S}$ of $k$ subtasks in \eqref{eq:structure} as follows: %over with its portions:% comprising the remaining time, actuator limits, and associated distances to achieve the subtasks in the sequence:
\vspace{-1mm}
\begin{equation}
\small
\label{eq:secondary_CBF}
\begin{split}
        \mathfrak{b}_i(\xs,t)=r&_{\boldsymbol{S}(i)}+ \frac{Dist(\xs,\mathcal{P}_{\boldsymbol{S}(1)})}{{u}_{\max}}\\& - \frac{\sum_{l=2}^{i} Dist(\mathcal{P}_{\boldsymbol{S}(l-1)},\mathcal{P}_{\boldsymbol{S}(l)})}{{u}_{\max}},\ \forall i\in \boldsymbol{S}\setminus\boldsymbol{S}(1),
        \end{split}
\end{equation}
\vspace{-4mm}

\noindent where $r_{\boldsymbol{S}(i)}$ is the remaining time to achieve $i^{th}$ subtask in the sequence $\boldsymbol{S}$. Nonnegativeness of $ \mathfrak{b}_i(\xs,t)$ for all $i$ simply implies that the system should achieve the first subtask in the order, $\Phi_{\boldsymbol{S}(1)}$, such that all the remaining subtasks can still be achieved on time. Therefore, in addition to staying inside the $\mathcal{C}$ (for subtask-specific achievement) we want our system to always be in $\mathcal{B}=\bigcap_i\mathcal{B}_i$ as well (for collective achievement) where $\mathcal{B}_i=\{\xs\in\mathbb{R}^n\ |\  \mathfrak{b}_i(\xs,t)\geq 0\}$. 

Note that the inner part $\varphi$ is in DNF (Def. \ref{def:DNF}).  The target set $\mathcal{P}_\varphi$ for each subtask in the sequence is determined as follows (unlike the primary CBFs which may contain multiple target sets). If $\varphi$ comprises conjunction, i.e., $\varphi=\bigwedge_i\varphi_i$, we select the target set assuming one subtask area contains another as $\mathcal{P}_\varphi=\arg\,\max_{\mathcal{P}_i}(Dist(\xs,\mathcal{P}_i))$. This assumption is not restrictive since we can always define a new compact target set by the intersection of the sets in conjunction. If the inner part $\varphi$ includes disjunction \big(e.g., $\varphi_1\vee(\varphi_2\wedge\varphi_3)$\big), we find multiple alternative sequences with each possible subtask with different $\varphi$'s (or conjunction of them) arising due to disjunction. Notice that this may increase the number of possible sequences to choose from exponentially when multiple subtasks are present containing disjunction operators. However, as will be discussed in Sec. \ref{sec:STL_CBF}, the sequence determination is implemented once and offline except the specifications requiring recurrence.

\vspace{-1mm}
\begin{continueexample}{example_1}
Consider the previous example with an initial state $x_0=3.5$. Let the set of alternative feasible sequences of subtasks $\Phi_1$ and $\Phi_2$ be $\{(\varphi_{11},\varphi_2),(\varphi_{12},\varphi_2)\}$ where $\varphi_{11}=x\geq2\wedge x\leq 3$, $\varphi_{12}=x\geq5 \wedge x \leq6$, and $\varphi_2=x\geq 7 \wedge x\leq8$. That is, one option is achieving $\varphi_{11}$ on time and then $\varphi_{2}$, and the other is satisfying $\varphi_{12}$ and $\varphi_{2}$, respectively. If we choose the second sequence $\boldsymbol{S}=(\varphi_{12},\varphi_2)$, in addition to the primary CBFs defined in the previous example, we have a secondary CBF enforcing: ``Satisfy $x\geq5 \wedge x \leq6$ so early that $x\geq 7 \wedge x\leq8$ can still be achieved on time." Then the secondary CBF constraint \eqref{eq:secondary_CBF} becomes $\mathfrak{b}(\xs,t)=r_{2}+ \frac{x-5}{{{u}}_{\max}}- \frac{ 7-5}{{{u}}_{\max}}\geq0$ which will be removed when $x=5$ is satisfied for the first time.
\end{continueexample}

\vspace{-1mm}
After determining the sequence of subtasks $\boldsymbol{S}$, we apply the secondary CBF constraint based on the most time-critical candidate in \eqref{eq:secondary_CBF} which is found as,
\begin{equation}
\vspace{-1.5mm}
\small
%\begin{split}
\label{eq:secondary_CBF_tightest}
        \mathfrak{b}(\xs,t)=\min_{i\in \boldsymbol{S}\setminus\boldsymbol{S}(1)}\; \;\mathfrak{b}_i(\xs,t).
        %s.t.&\; \xs=\xs_0,
        %\end{split}
        %\vspace{-1mm}
\end{equation}

\begin{remark}\label{rem:b_implies_h}
We use the most time-critical secondary CBF $\mathfrak{b}(\xs,t)\geq 0$ which implies $\mathfrak{b}_{i}(\xs,t)\geq 0,\ \forall i\in \boldsymbol{S}\setminus\boldsymbol{S}(1)$. Moreover, nonnegativeness of a candidate $\mathfrak{b}_i(\xs,t)\geq 0$ implies $\mathfrak{h}_{\boldsymbol{S}(i)}(\xs,t)\geq 0$ as well (i.e., $r_{\boldsymbol{S}(i)}\geq \mathfrak{d}_{\boldsymbol{S}(i)}$) since $r_{\boldsymbol{S}(i)}\geq \mathfrak{d}_{\boldsymbol{S}(1)} + \sum_{l=2}^i \mathfrak{d}_{\boldsymbol{S}(l-1),\boldsymbol{S}(l)}\geq \mathfrak{d}_{\boldsymbol{S}(i)}$. Therefore, $\mathfrak{h}_{\boldsymbol{S}(1)}(\xs,t)$ is the only primary CBF to be applied by containing the smooth robustness metrics of Boolean connectives as in \eqref{eq:robustness_degree} unlike the secondary CBF components (sequence terms) with definite target sets. Whenever $\Phi_{\boldsymbol{S}(1)}$ is achieved, the primary CBF constraint will switch to the next subtask in the sequence.\end{remark}
\vspace{-1mm}

%The CBF constraints that are constructed by the information of actuation limits and time bounds comprise a primary CBF \eqref{eq:primary_CBF} for the first subtask in the sequence $\Phi_{\boldsymbol{S}(1)}$ and a secondary CBF \eqref{eq:secondary_CBF_tightest} over the sequence $\boldsymbol{S}$. 
We define a quadratic cost for the system as \textcolor{black}{$\mathcal{J}(\boldsymbol{u}_t)=\boldsymbol{u}_t^T\boldsymbol{u}_t$}. Note that the CBF constraints in \eqref{eq:TVZCBF} are linear in control for a constant state. Therefore, if we could represent the STL satisfaction as CBF constraints, then the Problem \ref{problem_1} reduces to quadratic programming (QP) problems \textcolor{black}{with explicit input and CBF constraints} to be solved sequentially over the discrete time:
\begin{subequations} \label{eq:QP}
\small
\begin{align}
   \min_{\boldsymbol{u}\in\mathcal{U}}\; &\;\frac{1}{2}\boldsymbol{u}^T\boldsymbol{u}\\
    s.t.\; &\frac{\partial\mathfrak{h}_{\boldsymbol{S}(1)}}{\partial\xs}^T\boldsymbol{u}+\frac{\partial\mathfrak{h}_{\boldsymbol{S}(1)}}{\partial t}+\alpha_{\boldsymbol{S}(1)}(\mathfrak{h}_{\boldsymbol{S}(1)})\geq0, \label{eq:QP_a}\noeqref{eq:QP_a}\\
    &\frac{\partial\mathfrak{b}}{\partial\xs}^T\boldsymbol{u}+\frac{\partial\mathfrak{b}}{\partial t}+\alpha_{\mathfrak{b}}(\mathfrak{b})\geq0,\label{eq:QP_b}\noeqref{eq:QP_b}
\end{align}
\end{subequations}
%\begin{equation}
%\small
%    \begin{split}
%    \label{eq:QP}
%    \min_{\boldsymbol{u}\in\mathcal{U}}\; &\;\frac{1}{2}\boldsymbol{u}^T\boldsymbol{u}\\
%        s.t.\; &\frac{\partial\mathfrak{h}_{\boldsymbol{S}(1)}}{\partial\xs}^T\boldsymbol{u}+\frac{\partial\mathfrak{h}_{\boldsymbol{S}(1)}}{\partial t}+\alpha_{\boldsymbol{S}(1)}(\mathfrak{h}_{\boldsymbol{S}(1)})\geq0,\\
%        &\frac{\partial\mathfrak{b}}{\partial\xs}^T\boldsymbol{u}+\frac{\partial\mathfrak{b}}{\partial t}+\alpha_{\mathfrak{b}}(\mathfrak{b})\geq0,
%        \end{split}
%\end{equation}
where ${\partial (\cdot)}/{\partial t}=-{\partial (\cdot)}/{\partial r}$ since ${\partial r}/{\partial t}=-1$. Note that when a set of states associated with a subtask is visited for the first time, the term belongs to that subtask is removed from the sequence (therefore from the secondary CBF), and the primary CBF constraint \eqref{eq:QP_a} is switched to one that belongs to the next subtask in the sequence. In this regard, when there are no subtasks left in the sequence, we deduce that the STL specification $\Phi$ in \eqref{eq:structure} is satisfied.

\textcolor{black}{\noindent\textbf{Conservativeness of the approach:} The sequence of subtasks (Def. \ref{def:sequence}) is constructed via the ordered set distances (Def. \ref{def:set_distance}) which assume that a target set of states is reached through its farthest point with respect to the next target set. This may result in not having a feasible sequence $\boldsymbol{S}$ whereas the STL specification is actually achievable. The user may circumvent this by relaxing the secondary CBF constraint in \eqref{eq:QP_b}, and penalizing the relaxation variable in the objective function. 
%Second, in the definition of CBFs, we require the defined distances to be traversable by ${u}_{\max}$ \eqref{eq:input_set} for the sake of nonnegative CBFs in \eqref{eq:primary_CBF}, \eqref{eq:secondary_CBF_tightest}. This is conservative in the sense that the magnitude of the maximum resultant input vector may be greater than ${u}_{\max}$ in particular directions. Therefore, we ask for more time to reach target sets in such directions than needed while defining CBFs. 
After noting these, it is worth to reiterate that in this paper we assume STL specifications that have a feasible sequence $\boldsymbol{S}$, and CBF functions that are initially nonnegative (under ${u}_{\max}$) so that we could pursue forward invariance over their zero superlevel sets.}

\renewcommand{\arraystretch}{1.5}
\begin{table*}[t!]
\centering
\caption{Parametrization and implementation of the CBF constraints in accordance with the temporal operator of respective subtask and the associated time interval(s).}
\label{table:stl_cbf}
\begin{tabular}{|m{2.4cm}|m{1.6cm}|m{6cm}| m{6cm}|} 
 \hline
 Temporal Ops. of $\Phi_i$ & Initial Rem. Time, $r_i$ & Primary CBFs&Secondary CBF (Termwise)\\\hline
 $F_{[a,b]}\varphi_i$ & $r_i=b$& Once achieved within $[a,b]$, remove.& Remove with the associated Primary CBF.\\\hline $G_{[a,b]}\varphi_i$ & $r_i=a$& Once started, fix $r_i=\Delta t$, remove at $t=b$.& Remove at $t=a$.\\ \hline $F_{[a,b]}G_{[c,d]}\varphi_i$ & $r_i=b+c$& Once started at $t'\in[a+c,b+c]$ fix $r_i=\Delta t$, remove at $t=t'+d-c$.& Remove at $t=t'$.\\ \hline
 $G_{[a,b]}F_{[c,d]}\varphi_i$&$r_i=a+d$ & Each time achieved within $[a+c,b]$, reset $r_i=d$, once achieved within $[b+c,b+d]$, remove.&Each time $r_i$ is reset, reconstruct the sequence, remove with the associated Primary CBF.\\ 
 \hline
\end{tabular}
 \vspace{-4.5mm}
\end{table*}
\renewcommand{\arraystretch}{1}

\subsection{STL Control Synthesis under Dual CBF} \label{sec:STL_CBF}

We synthesize an STL controller using the primary CBFs in \eqref{eq:primary_CBF} for each subtask in \eqref{eq:structure} and a secondary CBF as in \eqref{eq:secondary_CBF} depicting a feasible sequence of subtasks, $\boldsymbol{S}$. In the QP problem \eqref{eq:QP}, we apply the primary CBF associated with the first subtask in the sequence which switches to the next when achieved. The secondary CBF, on the other hand, is constructed over the most time-critical portion of the sequence as in \eqref{eq:secondary_CBF_tightest}. The implementation of dual CBF can be summarized based on the type of the temporal operator(s) of each subtask $\Phi_i$ as follows: 
\begin{itemize}[leftmargin=*]
    \item If the operator is finally ($\Phi_i=F_{[a,b]}\varphi_i$), we apply the primary CBF and the associated term in the secondary CBF sequence until the inner subtask $\varphi_i$ is satisfied within the given time window $[a,b]$.
    \item If it is globally ($\Phi_i=G_{[a,b]}\varphi_i$) or finally-globally ($\Phi_i=F_{[a,b]}G_{[c,d]}\varphi_i$), once the subtask satisfaction begins, the associated term inside the secondary CBF is removed but we enforce the primary CBF constraint until the last time step $\varphi_i$ needs to hold by keeping the remaining time fixed to $r_i=\Delta t$, where $\Delta t$ is the time step between two consecutive solutions of \eqref{eq:QP}. Moreover, we add the length of time horizon (e.g., $b-a$ for $G_{[a,b]}\varphi_i$) to the next term in the sequence (if any) by subtracting the time system is supposed to spend in the target set for $\Phi_i$ while traversing it (due to the worst case set distance definition in Def. \ref{def:set_distance}). %We therefore account for the time the system must spend to hold globally in the sequence calculation. 
    Also note that this is the only case we may have two active primary CBFs: one for $\Phi_i$ with $r_i=\Delta t$ since it is removed from the sequence but is active within the globally time interval and one for the subtask next to $\Phi_i$ which becomes the first subtask in the sequence after $\Phi_i$ is removed. 
    \item If it is globally-finally ($\Phi_i=G_{[a,b]}F_{[c,d]}\varphi_i$), we treat it like finally, when the subtask is satisfied within the given time interval, we reset the remaining time and restart the process with a new sequence. 
    \item If it is until ($\Phi_i\!=\!\varphi_{i,1}U_{[a,b]}\varphi_{i,2}$) we use $\Phi_i=G_{[0,t']}\varphi_{i,1}\wedge F_{[a,b]}\varphi_{i,2}$ where $t'$ is the time of satisfaction of $\varphi_{i,2}$. 
    
\end{itemize}  
Details of the implementation can be found in Table \ref{table:stl_cbf}. %subtask-specific parameters such as remaining time and robustness metrics/distances \eqref{eq:robustness_degree} are the same for primary and secondary CBFs as long as the associated terms are not removed.

\begin{theorem}
Given an STL specification $\Phi=\Phi_1\wedge\Phi_2\wedge\ldots\wedge\Phi_k$ with the syntax in \eqref{eq:STL_fragment} and a sequence $\boldsymbol{S}$ of the $k$ subtasks as defined in Def. \ref{def:sequence} such that 
 \vspace{-3mm}
\begin{equation}
\vspace{-1mm}
\small
    r_{\boldsymbol{S}(j)}\geq \mathfrak{d}_{\boldsymbol{S}(1)} + \sum_{l=2}^j \mathfrak{d}_{\boldsymbol{S}(l-1),\boldsymbol{S}(l)},\ \forall j\in \boldsymbol{S},
\end{equation}
where $r_i$ is the remaining time to achieve $\Phi_i$;  $\mathfrak{d}_i$ and $\mathfrak{d}_{i,j}$ are defined in \eqref{eq:d_dist}. Then the control law obtained by sequentially solving \eqref{eq:QP} in every $\Delta t$ time units yields a trajectory for the system \eqref{eq:single_integrator} that satisfies $\Phi$.

%keeping the primary CBFs $\mathfrak{h}_i(\xs,t)\ \forall i\in\{1,\ldots,k\}$ \eqref{eq:primary_CBF},\eqref{eq:primary_CBF_OR} and secondary CBF \eqref{eq:secondary_CBF},\eqref{eq:secondary_CBF_tightest} nonnegative implies the satisfaction of $\Phi$. %\textcolor{red}{or Our implementation in Table 1 yields the problem to be feasible?}
\end{theorem}

\begin{proof}
Consider two time-varying CBFs nonnegative at time $t'$, e.g., $\mathfrak{h}(\xs_{t'},{t'})\geq0$ in \eqref{eq:primary_CBF} and $\mathfrak{b}(\xs_{t'},{t'})\geq0$ in \eqref{eq:secondary_CBF_tightest}, then by the Thm. 1 in \cite{lindemann2018control}, the control law $\boldsymbol{u}$ satisfying the constraints in \eqref{eq:QP_a} and \eqref{eq:QP_b} (i.e., $\boldsymbol{u}$ being inside the set $S_{\boldsymbol{u}}$ in \cite{lindemann2018control}) implies $\mathfrak{h}(\xs_t,t)\geq0$ and $\mathfrak{b}(\xs_t,t)\geq0,\ \forall t\geq t'$. Let us now show that it is feasible to keep them nonnegative. Let the target set be $\mathcal{P}$ and $\xs^*\in\partial \mathcal{P}$, then the initially nonnegative CBFs imply $r_{t'}u_{\max}\geq\|\xs^*-\xs_{t'}\|$ where $t'$ is the first time the target set is switched to $\mathcal{P}$. Note that the system \eqref{eq:single_integrator} is capable to move according to $\|\xs_t-\xs_{t'}\|=u_{\max}(r_{t'}-r_{t}),\ \forall t\geq t'$. Then we obtain $\|\xs_t-\xs_{t'}\|+r_{t} u_{\max}\geq\|\xs^*-\xs_{t'}\|$, i.e., at any time $t\in[t',t'+r_{t'}]$ the system can reach $\xs^*$ under $u_{\max}$, and it is feasible to sustain the nonnegativeness of the CBFs. The secondary CBF $\mathfrak{b}(\xs,t)$ is the most time-critical one among the candidates. %(Def. \ref{def:portion}) which does not change unless the associated subtask is achieved (Proposition \eqref{prop:tightest}).
Therefore, if it is nonnegative, then all the candidates in \eqref{eq:secondary_CBF} will be nonnegative as well. As discussed in the Remark \ref{rem:b_implies_h}, this implies the nonnegativeness of the primary CBFs except the one belongs to the $\Phi_{\boldsymbol{S}(1)}$ whose constraint is separately defined in \eqref{eq:QP_a}. Now let us consider the possible temporal operators:\\
\noindent \textbf{Case 1:} (Only $F_{[a,b]}\varphi$ among the subtasks). %Then the sequence $\boldsymbol{S}$ in the premise is updated by removing a term once a subtask, say $\Phi_{\boldsymbol{S}(j)}$, is satisfied. 
By Lemma \eqref{lem:triangle}, the new sequence $\boldsymbol{S}'=\boldsymbol{S}\setminus\boldsymbol{S}(j)$ obtained after a subtask $\Phi_{\boldsymbol{S}(j)}$ is satisfied and removed is also feasible as the secondary CBF remains nonnegative. Moreover, since $r_{i}$ decreases for all $i$ and $\mathfrak{b}(\xs,t)$ stays nonnegative (as well as $\mathfrak{b}_i(\xs,t)$'s and the associated primary CBFs as discussed in Remark \ref{rem:b_implies_h}), each subtask $\Phi_{i}$ is achieved and removed from $\boldsymbol{S}$ in  at most $r_{i}$ time units. This process continues until no subtask is left in the sequence $\boldsymbol{S}$.\\ 
\noindent \textbf{Case 2:} ($G_{[a,b]}F_{[c,d]}\varphi$). The subtask is treated as $F_{[c,d]}\varphi$ by resetting the remaining time once $\varphi$ is satisfied as explained in Table \ref{table:stl_cbf}. The guarantees given in Case 1 continues until the satisfaction of $\varphi$ that needs to be achieved repetitively. At each satisfaction, a new sequence is found by placing $F_{[c,d]}\varphi$ accordingly and applied as in Case 1.\\
%which again preserves the guarantees in the Case 1.
\noindent \textbf{Case 3:} ($G_{[a,b]}\varphi$ or $F_{[a,b]}G_{[c,d]}\varphi$). The remaining time for the globally operator is defined to be the initial time of its interval, e.g., $r_{G}=a$. Once the satisfaction starts, the remaining time is kept at $r_{G}=\Delta t$ until the end of globally time interval. This ensures the system to stay inside the target set of the subtask. Moreover, the absence of globally subtask in the secondary CBF while it is still active does not undermine the satisfaction since the time required to hold $\varphi$ is already considered in the sequence construction.
\end{proof}

\vspace{-4mm}
\begin{remark}\label{rem:sequence_selection}
A satisfying sequence $\boldsymbol{S}$ can be found (if any) by assessing the feasibility of all the possible orders of subtasks and making a selection based on some user defined criteria.  
\end{remark}

\textcolor{black}{In the simulations, we use the sequence that cumulatively has the highest difference between the remaining and required times among the alternatives feasible sequences.}

\section{Simulations}\label{sec:simulations}
We develop a control synthesis tool that generates trajectories satisfying the given STL specification \eqref{eq:structure} by solving \eqref{eq:QP} sequentially with the \texttt{quadprog} solver in MATLAB R2021a. A laptop computer with 1.8 GHz, Intel Core i5 processor is used to run the simulations.
\vspace{-2mm}

\color{black}\subsection{Benchmark Analysis}
\label{sec:benchmark_analysis}

Inclusion of actuation limits in the formulation of CBFs and considering the sequential satisfaction of STL subtasks have multiple advantages. We will show these benefits compared to the existing methods with the dynamics of $\dot{x}={u}\in\mathbb{R}$. Note that for each specification $\Phi$ in this subsection, we use $\Phi'$ in our method, which is obtained by bounding the predicates in $\Phi$ for the sake of compact target sets as our method requires. For example, in \eqref{eq:conflict_exmpl}, we apply $x\geq10\wedge x\leq11$ instead of $x\geq 10$, and $x\leq5\wedge x\geq4$ instead of $x\leq 5$. This neither undermines nor facilitates the accomplishment of the subtask (since $5\leq x_0\leq 10$). The original specifications are used for the other methods in the benchmark analysis as required.
% \Phi' implies \Phi but the reverse is not generally true!

First of all, \cite{lindemann2018control} mandates the achievement of fixed way points between the initial state and the target by using time-varying CBFs. Such an approach demands continuous progress toward the target and results in infeasibility when multiple targets are active at different directions. In \cite{xiao2021high}, such conflicting subtasks are tried to be handled via relaxations with finite-time convergent CBFs (causing delays in the satisfaction and potentially violation of the original specification). %While the authors' method may avoid infeasibility, the relaxation causes a delay in the satisfaction of relaxed subtask which yields a failure of the original STL specification. 
For example, consider an STL specification and its bounded equivalent under the initial condition $x_0=8$,
\begin{equation}
\small
\begin{split}
        \label{eq:conflict_exmpl}
    \Phi&=\Phi_1\wedge\Phi_2=F_{[0,5]}x\geq10\wedge F_{[1,6]}x\leq5,\\
    \Phi'\!&=F_{[0,5]}(x\geq10\wedge x\leq11)\wedge F_{[1,6]}(x\leq5\wedge x\geq4).
    \end{split}
\end{equation}

In \cite{xiao2021high}, the CBF associated with $\Phi_2=F_{[1,6]}x\geq5$ is relaxed since its time window is later than $\Phi_1=F_{[0,5]}x\geq10$. However, this relaxation yields the violation of $\Phi_2$ as Fig.~\ref{fig:conf} depicts. The dual CBF we propose, on the other hand, achieves both specifications on time by checking if it is feasible to go to first $x\leq5$, then $x\geq10$, and vice versa within the allowed time.  

\begin{figure}[htb!]
    \centering
    \includegraphics[trim={68 0 65 15},clip,width=.5\textwidth]{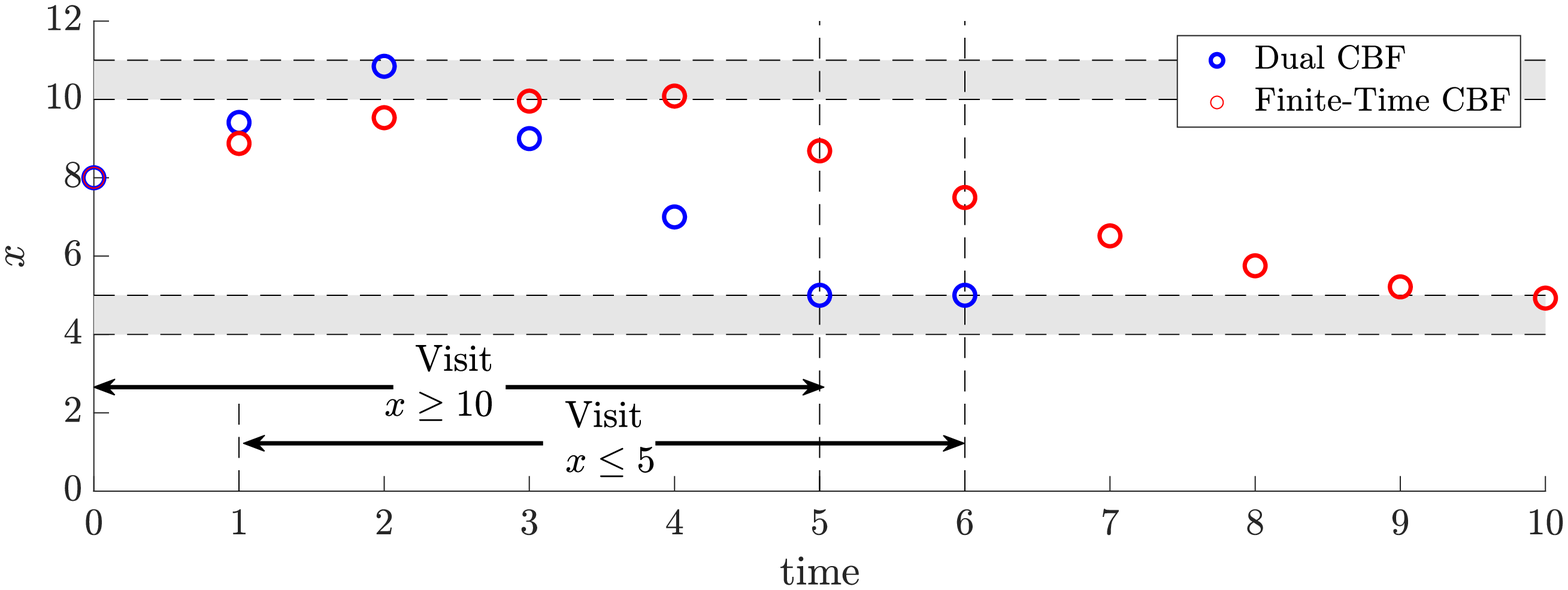}
    \caption{\textcolor{black}{Comparison of the two CBF approaches for the STL specifications $\Phi$ and $\Phi'$ in \eqref{eq:conflict_exmpl}.}}
    \label{fig:conf}
\end{figure}

The incapability to define and achieve conflicting specifications as in the above case results in failure for the periodic subtasks as well. Consider the STL specification and its bounded equivalent with $x_0=7$ given below,
\begin{equation}
\vspace{-.5mm}
\small
\begin{split}
    \label{eq:nested_exmpl}
    \Phi&=G_{[0,20]}F_{[0,10]}x\geq10 \wedge F_{[0,15]}x\leq5\wedge F_{[20,30]}x\leq3,\\
    \Phi'\!&=G_{[0,20]}F_{[0,10]}(x\geq10\wedge x\leq11) \wedge F_{[0,15]}(x\leq5\wedge x\geq4)\\
    &\hspace{30mm}\wedge F_{[20,30]}(x\leq3\wedge x\geq2).
    \end{split}
\end{equation}
    \vspace{-6mm}
\begin{figure}[htb!]
    \centering
    \includegraphics[trim={68 0 65 15},clip,width=.5\textwidth]{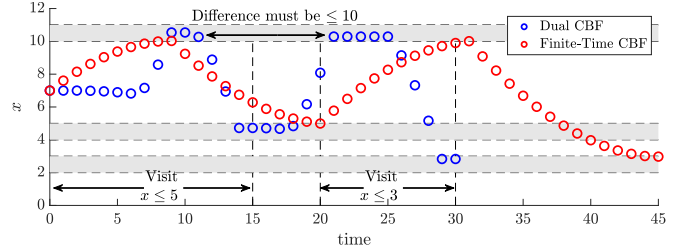}
    \caption{\textcolor{black}{Comparison of the two CBF approaches for the STL specifications $\Phi$ and $\Phi'$ in \eqref{eq:nested_exmpl}.}}
    \label{fig:nested}
    \vspace{-5mm}
\end{figure}

In general, recursive task definitions are missing in the STL CBF literature. Still, by giving the order of satisfaction for the predicates as a priori in an ad-hoc way, the approach in \cite{xiao2021high} may generate a trajectory that visits desired regions repetitively. But it fails to do so on time again even when the relaxation is applied (and the order of relaxations is predetermined). Figure~\ref{fig:nested} depicts this case while the dual CBF approach we propose generates a trajectory that repetitively visits $x=10$ while achieving other subtasks on time. Another advantage of the proposed dual CBF approach is better handling the disjunction operator. %\textcolor{red}{Another advantage of the proposed dual CBF approach is yielding more efficient system trajectories in terms of early satisfaction (can also be seen in in Fig.~\ref{fig:conf}) and input cost.} 
To illustrate, consider the below STL specification with its equivalent under $x_0=5$,
\begin{equation}
\small
\begin{split}
        \label{eq:disj_exmpl}
    \Phi&=\Phi_1\wedge\Phi_2=F_{[10,15]}x\geq9\wedge\big(F_{[0,5]}x\leq3\vee F_{[0,5]}x\geq7.5\big),\\
    \Phi'\!&=F_{[10,15]}(x\geq9\wedge x\leq10)\wedge\big(F_{[0,5]}(x\leq3\wedge x\geq2)\\
    &\hspace{45mm}\vee F_{[0,5]}(x\geq7.5\wedge x\leq8.5)\big).
    \end{split}
\end{equation} 
    \vspace{-4mm}
    
As Fig.~\ref{fig:disj} represents, for \cite{lindemann2018control}\footnote{The work in \cite{lindemann2018control} normally does not consider disjunction operator and preserves soundness. Hence, we added a sound disjunction operator (by using the smooth approximation of $\max$ in Def. \ref{def:smooth_approx}) to enrich the comparison.} and \cite{xiao2021high}, applying the disjunction yields to achieve the closer alternative. However, this may cause an unnecessary delay to accomplish other subtasks (e.g., $\Phi_1$ in \eqref{eq:disj_exmpl}) and an additional input cost.

\begin{figure}[htb!]
\vspace{-4mm}
    \centering
    \includegraphics[trim={68 0 65 10},clip,width=.5\textwidth]{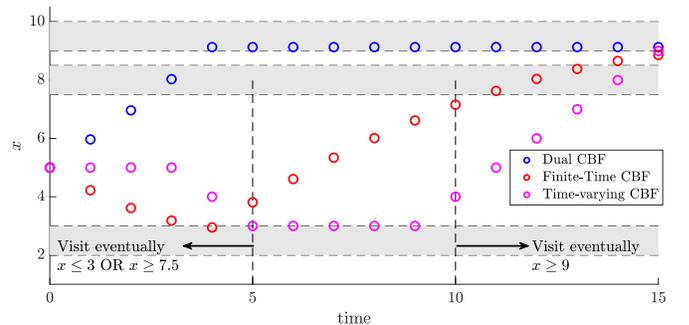}
    \caption{\textcolor{black}{Comparison of the three CBF approaches for the STL specifications $\Phi$ and $\Phi'$ in \eqref{eq:disj_exmpl}.}}
    \label{fig:disj}
    \vspace{-4mm}
\end{figure}

    \vspace{-2mm}
\subsection{Case Study}\label{sec:case_study}
\color{black}
We \textcolor{black}{also} consider a more complex scenario to illustrate the capabilities of the proposed approach\textcolor{black}{, which cannot be defined using other STL CBF approaches (e.g. \cite{lindemann2018control,xiao2021high}),} with $\xs=[x\ y]^T$, $\boldsymbol{u}=[{u}_x\ {u}_y]^T$, and $\|\boldsymbol{u}\|\leq1$. The STL specification the system is required to achieve is 
\begin{equation}\label{eq:case_study_complex}
\small
\begin{split}
\Phi&=\Phi_1\wedge\Phi_2\wedge\Phi_3\wedge\Phi_4\wedge\Phi_5\\
&=G_{[0,15]}F_{[0,10]}\ Region_1 \wedge F_{[0,15]}\ Region_2 \wedge F_{[0,15]}\ Region_3 \\ 
&\wedge F_{[0,40]}G_{[0,10]}\ Region_4 \wedge G_{[38,45]}\ \big(Region_5\vee Region_{6}\big),
\end{split}
\end{equation}
\vspace{-3mm}

\noindent where $Region_i=\{x,y\in\mathbb{R}\ |\ (x-\boldsymbol{x'}(i))^2+(y-\boldsymbol{y'}(i))^2\leq \boldsymbol{rad}(i)^2\}$, and we use $\boldsymbol{rad}=[1.5,0.5,1,1,0.75,0.75]$, $\boldsymbol{x'}=[7,2,12,7,11,3]$, and $\boldsymbol{y'}=[6,5,5,2,2,2]$ for $i=1,\ldots,6$. Note that the subtasks defined on the regions $1-4$ have conflicts, and the last two substasks may have overlapping requirements as well (depending on when region $4$ is visited).

\begin{figure}[htb!]
%\vspace{-2mm}
    \centering
    \includegraphics[trim={38 6 40 20},clip,width=.5\textwidth]{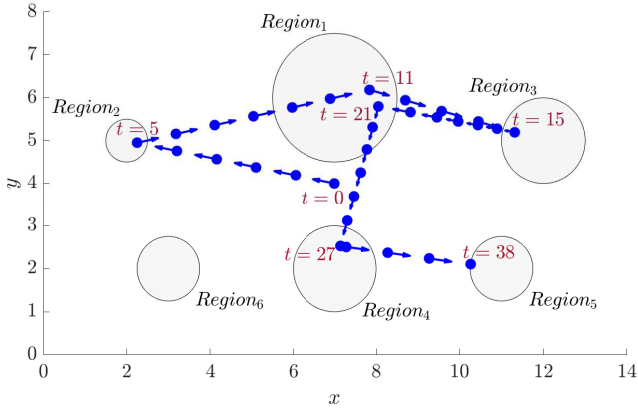}
    \caption{Trajectory of the system satisfying the STL specification in \eqref{eq:case_study_complex}.}
    \label{fig:case_study_complex}
\end{figure}

A new sequence is calculated only once at $t=11$. \textcolor{black}{Moreover, while region $6$ is listed in the new sequence as the last one to be visited, the primary CBF corresponding to $\Phi_5$ in \eqref{eq:case_study_complex} causes the system to visit region $5$ instead as  Fig.~\ref{fig:case_study_complex} depicts. This is because of the robustness metrics inside the primary CBF which accommodates the redundancy provided by the disjunction operator.} %This illustrates how our method preserves the redundancy the STL provides. 
\color{black}The primary and secondary CBF values that are required to be nonnegative throughout the mission are shown in Fig.~\ref{fig:cbf_values}. Note that while the primary CBF's are switched as the substasks are accomplished, the secondary CBF only loses term as the regions are visited.

\begin{figure}[htb!]
    \centering
    \includegraphics[trim={68 7 65 13},clip,width=.5\textwidth]{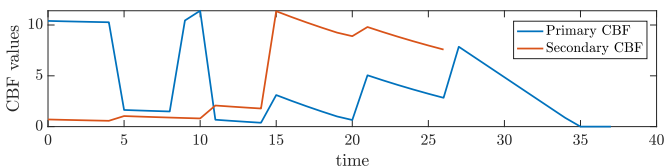}
    \caption{\textcolor{black}{The change in the values of primary \eqref{eq:primary_CBF} and secondary \eqref{eq:secondary_CBF_tightest} control barrier functions which are required to be always nonnegative.}}
    \label{fig:cbf_values}
\end{figure}
\color{black}

Finally, we compare our results for the STL specification in \eqref{eq:case_study_complex} with the common STL control synthesis methods that include the solution of a mixed-integer quadratic program\footnote{Linearity requirement in the constraints of the MIQP problem is met by specifying the subtask regions as inner-fitted squares to the circular subtask areas in the original specification \eqref{eq:case_study_complex}.} (MIQP) as in \cite{raman} and a nonlinear program (NLP) formulated with the smooth STL robustness metrics \cite{pant2017smooth}. Table \ref{tab:results} shows that our approach provides a time-efficient solution for a reasonable cost of additional input. Furthermore, compared to the MIQP and NLP approaches which are defined through a horizon, the sequential implementation of the QP in \eqref{eq:QP} facilitates the real-time applicability with the stepwise solutions in order of milliseconds.

\renewcommand{\arraystretch}{1.14}%1.35
\begin{table}[htb!]
\vspace{-2mm}
    \centering
        \caption{Comparison of the results with the popular control synthesis approaches for the STL specification in \eqref{eq:case_study_complex}.}
\begin{tabular}{c||c|c|c}
 Solution Method& \begin{tabular}{@{}c@{}}Dual CBF\end{tabular} & \begin{tabular}{@{}c@{}}Smooth Rob.\\ Metrics (NLP)\end{tabular} & \begin{tabular}{@{}c@{}} MIQP Encoding\end{tabular}\\\hline
    \begin{tabular}{@{}c@{}}Total Cost\\ $\sum_t\mathcal{J}
       ^*$ $[-]$\end{tabular} &$20.29$ &$14.65$ &$17.03$\\\hline
     \begin{tabular}{@{}c@{}}Solution\\ Time $[s]$\end{tabular} &\begin{tabular}{@{}c@{}}$0.48$\\(Avg.\\ $0.013\ ms$)\end{tabular} &$2.29$ &$9.24$ \\
    \end{tabular}
    \label{tab:results}
    %\vspace{-1mm}
\end{table}
\renewcommand{\arraystretch}{1}
\vspace{-6mm}
\section{Conclusion}\label{sec:conclusion}
In this paper, we proposed novel CBF formulations to satisfy a rich family of STL specifications accommodating nested temporal operators, specifications with conflicting requirements, and Boolean connectives inside the temporal operators. The proposed dual CBF formulations include the actuation limits and the satisfiability of multiple subtasks in the mission. Consequently, the time-varying feasible sets the system has to reside inside are constructed by the worst-case scenarios (e.g., travelling under maximum input policy). This provides redundancy to the system in achieving subtasks and yields to satisfy a rich family of specifications in a computationally efficient manner.

%\textcolor{red}{As a future direction, we plan to extend this work to multi-agent settings with possible couplings among the system. Another extension would be the robust satisfaction of given specifications. Finally, applicability of the newly proposed CBF formulations for the more complex systems, e.g., the systems with high relative degree, will be investigated.}

% a more sophisticated framework for sequence construction and selection
\bibliography{arxiv_submission}
\vspace{-6mm}

\end{document}